\documentclass[a4paper,english]{lipics-v2021}

\usepackage{hyperref}
\usepackage{comment}

\usepackage{stfloats}

\usepackage{graphicx} 

\usepackage{amsmath,amsfonts,amssymb,xspace}
\usepackage{graphicx}
\usepackage{paralist}
\usepackage{xcolor}

\usepackage{algorithm}
\usepackage[noend]{algpseudocode}
\usepackage{xifthen}
\usepackage{float}
\usepackage{adjustbox}

\usepackage{tikz}
\usepackage{pgfplots}

\usepackage{xcolor}
\usepackage{xparse}
\usepackage{xargs}

\usepackage{multicol}
\pgfplotsset{compat=1.17}
\DeclareUnicodeCharacter{2212}{−}
\usepgfplotslibrary{groupplots,dateplot}
\usetikzlibrary{patterns,shapes.arrows}
\usepackage{pgfplots}
\usepgfplotslibrary{colorbrewer}
\pgfplotsset{compat = 1.15, cycle list/Set1-8} 
\usetikzlibrary{pgfplots.statistics, pgfplots.colorbrewer} 
\pgfplotsset{compat=1.18}
\usepackage{pgfplotstable}

\makeatletter
\newcommand\resetstackedplots{
\makeatletter
\pgfplots@stacked@isfirstplottrue
\makeatother
\addplot [forget plot,draw=none] coordinates{(1,0) (2,0) (3, 0) (4, 0) (5, 0) (6, 0) (7, 0) (8, 0) (9, 0) };
}
\makeatother

\makeatletter
\newcommand\resetstackedplotsb{
\makeatletter
\pgfplots@stacked@isfirstplottrue
\makeatother
\addplot [forget plot,draw=none] coordinates{(10000, 0.0001) (20000, 0.0001) (30000, 0.0001) (40000, 0.0001) (50000, 0.0001) (60000, 0.0001) (70000, 0.0001) (80000, 0.0001) (90000, 0.0001) (100000, 0.0001) (200000, 0.0001) (300000, 0.0001) (400000, 0.0001) (500000, 0.0001) (600000, 0.0001) (700000, 0.0001) (800000, 0.0001) (900000, 0.0001) (1000000, 0.0001) (2000000, 0.0001)};
}
\makeatother

\definecolor{bgmm}{RGB}{237, 225, 97}
\definecolor{dgmm}{RGB}{199, 177, 36}
\definecolor{bhaca}{RGB}{247, 118, 109}
\definecolor{dhaca}{RGB}{214, 62, 51}
\definecolor{bkl15}{RGB}{78, 179, 64}
\definecolor{dkl15}{RGB}{56, 156, 42}
\definecolor{bkl10}{RGB}{140, 222, 129}
\definecolor{dkl10}{RGB}{101, 184, 90}
\definecolor{bkl5}{RGB}{186, 242, 179}
\definecolor{dkl5}{RGB}{168, 232, 160}
\definecolor{btmm}{RGB}{237, 140, 216}
\definecolor{dtmm}{RGB}{161, 72, 142}
\definecolor{baca}{RGB}{146, 171, 224}
\definecolor{daca}{RGB}{74, 108, 181}

\def\eps{\varepsilon}
\def\bR{{\mathbb R}}
\def\bN{{\mathbb N}}

\def\bX{{\mathbb X}}
\def\cD{\mathcal{D}}
\def\km{{\mathrm{km}}}

\def\RSpace{{\pazocal{R}}}

\DeclareMathAlphabet{\pazocal}{OMS}{zplm}{m}{n}

\renewcommand{\emph}[1]{\textit{\textbf{#1}}}

\newcommand{\disk}[2][\Delta]{
    \ifthenelse{\isempty{#2}}
    {\mathrm{b}_{#1}}
    {\mathrm{b}_{#1}(#2)}
}
\newcommand{\mdisk}[3][\Delta]{
    \ifthenelse{\isempty{#3}}
    {\mathrm{b}_{#1}^{#2}}
    {\mathrm{b}_{#1}^{#2}(#3)}
}

\newcommand{\dfree}[3][\Delta]{%
    \ifthenelse{\equal{#2}{} }
    {\cD_{#1}}
    {\cD^{(#2,#3)}_{#1}}
}

\newcommand{\df}{\mathrm{d}_\mathcal{F}}
\newcommand{\Cov}{\mathrm{Cov}}

\title{Finding Complex Patterns in Trajectory Data via Geometric Set Cover\footnote{Compared to the first version of this manuscript that appeared on arXiv, this is an updated version with more extensive experiments.}}
\titlerunning{Finding Complex Patterns in Trajectory Data}
\author{Jacobus Conradi}{University of Bonn, Germany}{conradi@cs.uni-bonn.de}{https://orcid.org/0000-0002-8259-1187}{Partially funded by the Deutsche Forschungsgemeinschaft (DFG, German Research Foundation) - AA 1111/2-2 (FOR 2535 Anticipating Human Behavior) and the iBehave Network: Sponsored by the Ministry of Culture and Science of the State of North Rhine-Westphalia.}
\author{Anne Driemel}{University of Bonn, Germany}{driemel@cs.uni-bonn.de}{https://orcid.org/0000-0002-1943-2589}{Affiliated with Lamarr Institute for Machine Learning and Artificial Intelligence.}

\authorrunning{J. Conradi and A. Driemel}

\Copyright{Jacobus Conradi, Anne Driemel} 
\ccsdesc[500]{ Theory of computation ~ Design and analysis of algorithms}

\keywords{Clustering, Set cover, Fr\'echet distance, Approximation algorithms}

\ArticleNo{100}
\begin{document}
\nolinenumbers
\maketitle

\newcommandx*\drawChart[4][1=linear, 4=0.05cm]{
    \begin{tikzpicture}
    \begin{axis}[
            width=12cm,
            height=8cm,
            ybar=2*\pgflinewidth,
            bar width=#4,
            ymin=0,
            ylabel={#2},
            xlabel={TAG},
            log origin=infty,
            ymode= #1,
            xtick={1,2,...,14},
            xticklabels={1,2,3,4,5,6,7,8,9,10,11,12,13,14},
            legend columns=4,
            column sep=1ex,
            ymax=1.19,
            enlarge x limits=0.07,
            grid=major,
            title={#2 Comparison}
            ]
            
            \addplot[
                ybar,
                fill=bkl5,
                draw=dkl5
            ] table[x=TAG,y=KlCluster 5,col sep=comma] {#3};
            \addlegendentry{KlCluster 5}
    
            \addplot[
                ybar,
                fill=bkl10,
                draw=dkl10
            ] table[x=TAG,y=KlCluster 10,col sep=comma] {#3};
            \addlegendentry{KlCluster 10}
    
            \addplot[
                ybar,
                fill=bkl15,
                draw=dkl15
            ] table[x=TAG,y=KlCluster 15,col sep=comma] {#3};
            \addlegendentry{KlCluster 15}
    
            \addplot[
                ybar,
                fill=btmm,
                draw=dtmm
            ] table[x=TAG,y=Tmm,col sep=comma] {#3};
            \addlegendentry{Tmm}
    
            \addplot[
                ybar,
                fill=baca,
                draw=daca
            ] table[x=TAG,y=Aca,col sep=comma] {#3};
            \addlegendentry{Aca}
    
            \addplot[
                ybar,
                fill=bhaca,
                draw=dhaca
            ] table[x=TAG,y=Haca,col sep=comma] {#3};
            \addlegendentry{Haca}
    
            \addplot[
                ybar,
                fill=bgmm,
                draw=dgmm
            ] table[x=TAG,y=Gmm,col sep=comma] {#3};
            \addlegendentry{Gmm}
    
        \end{axis}
    \end{tikzpicture}
}
\newcommandx*\drawChartB[5][1=linear, 4=0.05cm, 5=1.19]{
    \begin{tikzpicture}
    \begin{axis}[
            width=1.6\textwidth,
            height=80mm,
            ybar=2*\pgflinewidth,
            bar width=#4,
            ymin=0.1,
            ylabel={\small #2},
            log origin=infty,
            ymode= #1,
		      ymajorgrids,
            ymax= #5,
            xtick={1,2,...,14},
		      xticklabel style = {font=\footnotesize},
		      xticklabels = {\texttt{1},\texttt{2},\texttt{3},\texttt{4},\texttt{5},\texttt{6},\texttt{7},\texttt{8},\texttt{9},\texttt{10},\texttt{11},\texttt{12},\texttt{13},\texttt{14}},
		      xtick style = {draw=none},
            yticklabel style = {font=\small},
            legend columns=4,
            column sep=1ex,
            enlarge x limits=0.07,
            ]
            
            \addplot[
                ybar,
                fill=bkl5,
                draw=dkl5
            ] table[x=TAG,y=KlCluster 5,col sep=comma] {#3};
            \addlegendentry{our ($l=5$)}
    
            \addplot[
                ybar,
                fill=bkl10,
                draw=dkl10
            ] table[x=TAG,y=KlCluster 10,col sep=comma] {#3};
            \addlegendentry{our ($l=10$)}
    
            \addplot[
                ybar,
                fill=bkl15,
                draw=dkl15
            ] table[x=TAG,y=KlCluster 15,col sep=comma] {#3};
            \addlegendentry{our ($l=15$)}
    
            \addplot[
                ybar,
                fill=btmm,
                draw=dtmm
            ] table[x=TAG,y=Tmm,col sep=comma] {#3};
            \addlegendentry{TS}
        
            \addplot[
                ybar,
                fill=bhaca,
                draw=dhaca
            ] table[x=TAG,y=Haca,col sep=comma] {#3};
            \addlegendentry{HACA}
            
            \addplot[
                ybar,
                fill=baca,
                draw=daca
            ] table[x=TAG,y=Aca,col sep=comma] {#3};
            \addlegendentry{ACA}
    
            \addplot[
                ybar,
                fill=bgmm,
                draw=dgmm
            ] table[x=TAG,y=Gmm,col sep=comma] {#3};
            \addlegendentry{SC}
    
        \end{axis}
    \end{tikzpicture}
}
\newcommandx*\drawChartStack[5][1=linear, 5=0.05cm]{
    \begin{tikzpicture}
    \begin{axis}[
            width=12cm,
            height=8cm,
            ybar stacked, 
            bar width=#5,
            ymin=0,
            ylabel={#2},
            xlabel={#3},
            log origin=infty,
            ymode= #1,
            xtick={0.2,0.3,0.4,...,1.0},
            xticklabels={0.2,0.3,0.4,0.5,0.6,0.7,0.8,0.9,1.0},
            xmax=1,
            legend style={at={(1,1.05)}, anchor=south east},
            enlarge x limits=0.07,
            grid=major,
            title={#2 Comparison},
            unbounded coords=jump
            ]
            
            \addplot[
                ybar,
                fill=dtmm,
                draw=dtmm,
                bar shift=-3*#5
            ] table[x=Simp Delta,y=SimplifyTime KlCluster 5,col sep=comma] {#4};
    
            \addplot[
                ybar,
                fill=btmm,
                draw=btmm,
                bar shift=-3*#5
            ] table[x=Simp Delta,y=SolvingTime KlCluster 5,col sep=comma] {#4};
            \addlegendentry{KlCluster 5}

             \resetstackedplots

            \addplot[
                ybar,
                fill=dhaca,
                draw=dhaca,
            ] table[x=Simp Delta,y=SimplifyTime KlCluster 10,col sep=comma] {#4};
           \addplot[
                ybar,
                fill=bhaca,
                draw=bhaca
            ] table[x=Simp Delta,y=SolvingTime KlCluster 10,col sep=comma] {#4};
            \addlegendentry{KlCluster 10}

            \resetstackedplots

            \addplot[
                ybar,
                fill=daca,
                draw=daca,
                bar shift=3*#5,
            ] table[x=Simp Delta,y=SimplifyTime KlCluster 15,col sep=comma] {#4};
           \addplot[
                ybar,
                fill=baca,
                draw=baca,
                bar shift=3*#5
            ] table[x=Simp Delta,y=SolvingTime KlCluster 15,col sep=comma] {#4};
            \addlegendentry{KlCluster 15}

        \end{axis}
    \end{tikzpicture}
}

\newcommand{\boxplot}[2]{
\begin{tikzpicture}[]
	\pgfplotstableread[col sep=comma]{#1}\csvdata
 
	\begin{axis}[
		boxplot/draw direction = y,
		axis x line* = bottom,
		axis y line = left,
		enlarge y limits,
        enlarge x limits = 0.05,
		ymajorgrids,
        height=0.5\textwidth,
        width=0.5\textwidth,
		xtick = {1, ..., 7},
		xticklabel style = {anchor=east,rotate=45, font=\footnotesize},
		xticklabels = {our ($l=5$) ,our ($l=10$) ,our ($l=15$) , TS , ACA , HACA , SC },
		xtick style = {draw=none},
		ylabel = {#2},
		ytick = {0.2, 0.4, 0.6, 0.8, 1.0},
        ymax=1,
        ymin=0.1,
        boxplot/box extend=0.3
	]

		\addplot+[boxplot, fill=bkl5, draw=black] table[y expr=\thisrow{KlCluster 5}] {\csvdata}; 
        \addplot+[only marks, mark=x, mark size=1pt, mark options={black, fill opacity=0.5}] table[y expr=\thisrow{KlCluster 5}, x expr=1] {\csvdata};
  
        \addplot+[boxplot, fill=bkl10, draw=black] table[y expr=\thisrow{KlCluster 10}] {\csvdata}; 
        \addplot+[only marks, mark=*, mark size=1pt, mark options={black, fill opacity=0.5}] table[y expr=\thisrow{KlCluster 10}, x expr=2] {\csvdata};

        \addplot+[boxplot, fill=bkl15, draw=black] table[y expr=\thisrow{KlCluster 15}] {\csvdata}; 
        \addplot+[only marks, mark=square, mark size=1pt, mark options={black, fill opacity=0.5}] table[y expr=\thisrow{KlCluster 15}, x expr=3] {\csvdata};
        
        \addplot+[boxplot, fill=btmm, draw=black] table[y expr=\thisrow{Tmm}] {\csvdata}; 
        \addplot+[only marks, mark=triangle, mark size=1pt, mark options={black, fill opacity=0.5}] table[y expr=\thisrow{Tmm}, x expr=4] {\csvdata};

        \addplot+[boxplot, fill=bhaca, draw=black] table[y expr=\thisrow{Haca}] {\csvdata}; 
        \addplot+[only marks, mark=diamond, mark size=1pt, mark options={black, fill opacity=0.5}] table[y expr=\thisrow{Haca}, x expr=6] {\csvdata};
        
        \addplot+[boxplot, fill=baca, draw=black] table[y expr=\thisrow{Aca}] {\csvdata}; 
        \addplot+[only marks, mark=triangle*, mark size=1pt, mark options={black, fill opacity=0.5}] table[y expr=\thisrow{Aca}, x expr=5] {\csvdata};

        \addplot+[boxplot, fill=bgmm, draw=black] table[y expr=\thisrow{Gmm}] {\csvdata}; 
        \addplot+[only marks, mark=star, mark size=1
        pt, mark options={black, fill opacity=0.5}] table[y expr=\thisrow{Gmm}, x expr=7] {\csvdata};

	\end{axis}
\end{tikzpicture}
}







\begin{abstract}
    Clustering trajectories is a central challenge when faced with large amounts of movement data such as GPS data. We study a clustering problem that can be stated as a geometric set cover problem: Given a polygonal curve of complexity $n$, find the smallest number $k$ of representative trajectories of complexity at most $l$ such that any point on the input trajectories lies on a subtrajectory of the input that has Fréchet distance at most $\Delta$ to one of the representative trajectories. 
    In previous work, Brüning et al.~(2022) developed a bicriteria approximation algorithm that returns a set of curves of size $O(kl\log(kl))$ which covers the input with a radius of $11\Delta$ in time $\Tilde{O}((kl)^2n + kln^3)$, where $k$ is the smallest number of curves of complexity $l$ needed to cover the input with a radius of $\Delta$. The representative trajectories computed by this algorithm are always line segments. In the applications however, one is usually interested in more complex representative curves which consist of several edges. We present a new approach that builds upon previous work computing a set of curves of size $O(k\log(n))$ in time $\Tilde{O}(l^2n^4 + kln^4)$ with the same distance guarantee of $11\Delta$, where each curve may consist of curves of complexity up to the given complexity parameter~$l$. We conduct experiments on tracking data of ocean currents and full body motion data suggesting its validity as a tool for analyzing large spatio-temporal data sets.
\end{abstract}

\section{Introduction}
Advancements in motion tracking technology made it possible to observe and track spatio-temporal phenomena from many different areas affected by climate change ranging from ocean currents to animal migration. The authors of \cite{wilson2016climate} observed changes in ocean currents driven by climate change and analysed how these affect the dispersal of marine life providing evidence of the importance of understanding and predicting these current changes. Similar effects can be observed among migratory land-bound animals, as temperature and resource availability changes \cite{kubelka2022animal} and displacement of human life \cite{tabe2019climate}. Practitioners are often confronted with vast amounts of data from which one would like to extract a reoccuring pattern, preferably of small complexity, making the data more accessible to less efficient algorithms or schematic visualization. 
Identifying such patterns is a particular challenge as the type of pattern which is sought after may vary depending on the data and specific application. Since the quality assessment of patterns varies, there are multiple approaches to tackle the problem of subtrajectory clustering ranging from  heuristics~\cite{LeeHW07} to machine learning approaches, such as reinforcement learning~\cite{liang2024sub}, and further to formal approaches utilising combinatorial optimisation and similar algorithmic techniques (see the survey papers ~\cite{BuchinW20, wang2021survey, yuan2017review}). A popular quality measure often used in this context is the Fréchet distance. This is a distance measure defined on the space of curves that contains all trajectories. It was used in the work of Agarwal et al.~\cite{agarwal2018} and that of Buchin et al.~\cite{buchinGroup20}, among others. The approach we want to focus on is that of Akitaya et al.~\cite{Akitaya2021Covering}. They pose the problem as a geometric set cover problem, in which a given trajectory needs to be ``covered'' by the smallest possible number of ``center'' trajectores, such that each point of the input trajectory is contained in a subtrajectory of the input trajectory which has a small Fréchet distance to one of the center trajectories. This can similarly be thought of as a clustering problem in which each point on the input trajectory is assigned to at least one center trajectory. One drawback for practical applications of the approaches presented in \cite{Akitaya2021Covering} and the subsequent work of Brüning et al.~\cite{Brüning2022Faster} is that the center trajectories computed only ever consist of a single edge, wheras in applications one is often interested in finding center trajectories of higher complexity, as even a simple circular motion (as is present in gulf streams for example) cannot be modeled well with a single edge. In this paper we focus on extending the approach of~\cite{Brüning2022Faster} to allow center trajectories of higher complexity. We validate our approach by conducting experiments with real data from two different application areas.

\begin{figure}
    \centering
    \includegraphics[width=0.8\linewidth]{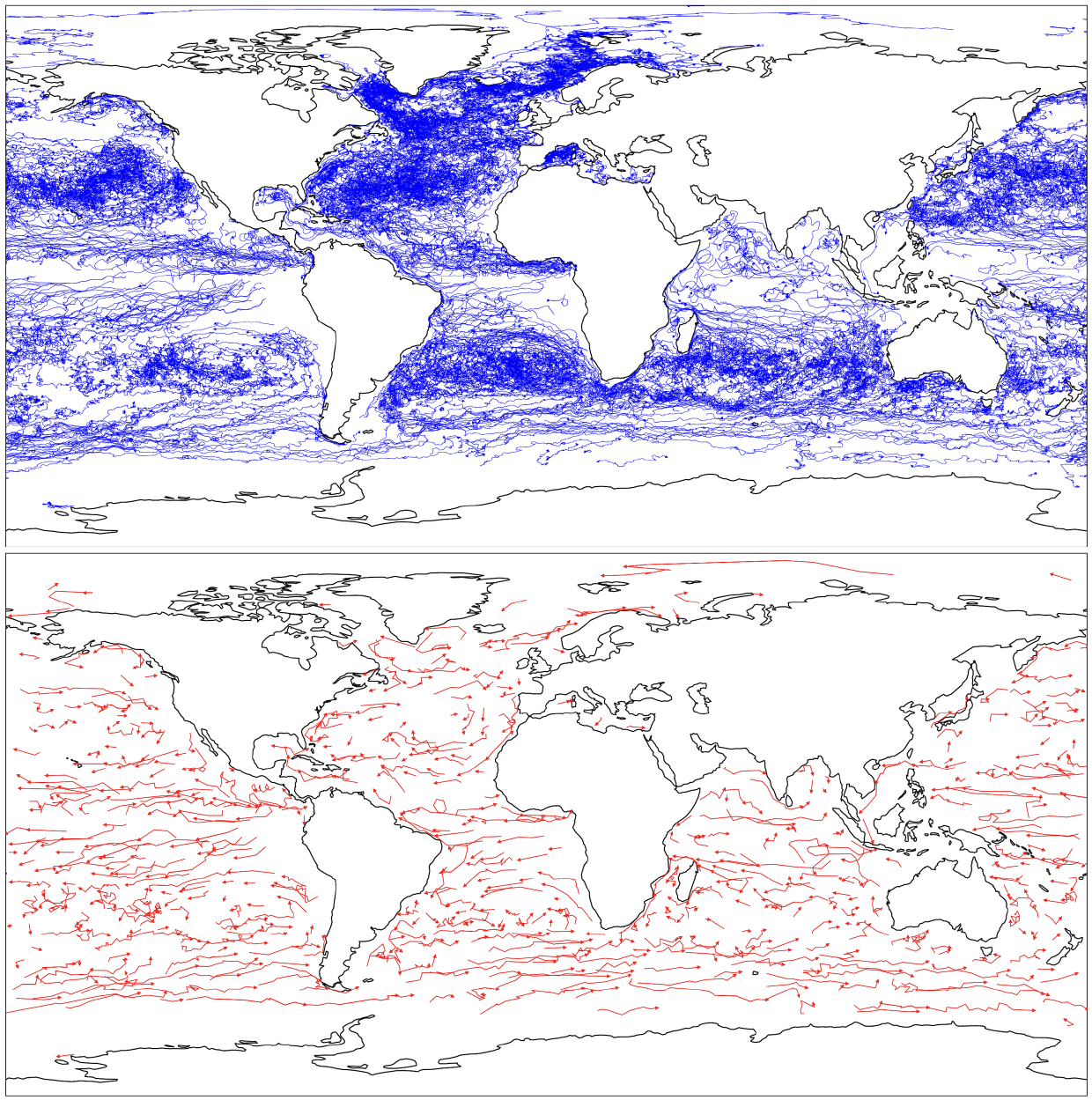}
    \caption{Illustration of $\approx2000$ individual ocean surface drifters and a resulting clustering.}
    \label{fig:alldrifters}
\end{figure}

\begin{figure*}
    \centering
    \includegraphics[width=0.75\textwidth]{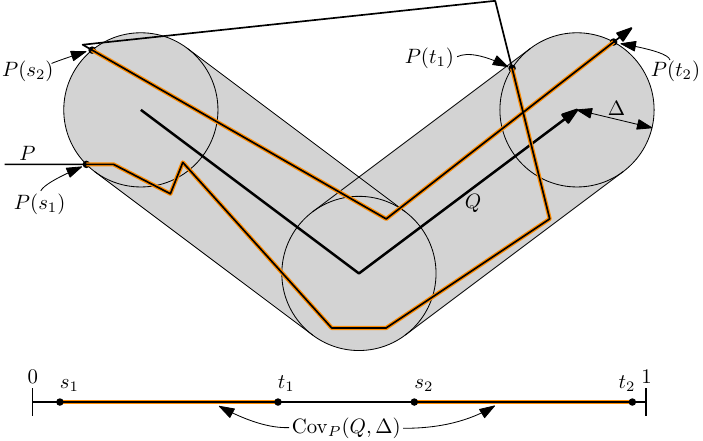}
    \caption{Illustration of the $\Delta$-coverage of $Q$ on the curve $P$.}
    \label{fig:coverage}
\end{figure*}

\subsection{Preliminaries}
A \emph{polygonal curve} $P$ in $\bR^d$ of \emph{complexity} $n$ is defined by an ordered set of points $(v_0,\ldots,v_n)\subset\bR^d$ by concatenating the linear interpolations of consecutive points. That is for each $1\leq i\leq n$ we obtain the \emph{edge} $e_i(t)=(1-t)v_{i-1}+tv_i$ and their concatenation $e_1\oplus\ldots\oplus e_n:[0,1]\rightarrow\bR^d$ defines $P$. We may denote an edge from $p$ to $q$ by $\overline{p\,q}$. We denote the set of all polygonal curves in $\bR^d$ of complexity at most $l\in\bN$ by $\bX^d_l$ and the complexity of a polygonal curve $P$ by $|P|$.
For a polygonal curve $P$ and given $0\leq s\leq t\leq 1$ we denote the \emph{subcurve} of $P$ from $P(s)$ to $P(t)$ by $P[s,t]$. If we drop the requirement that $s\leq t$, we say $P[s,t]$ is a \emph{free subcurve} of $P$. A free subcurve of $P$ is either a subcurve of $P$ or a subcurve of $P$ parametrized in the reverse direction.

For two curves in $\bR^d$ their continuous Fréchet distance is defined as
\[\df(P,Q)=\inf_{f,g}\max_{t\in[0,1]}\|P(f(t))-Q(g(t))\|\]
where $f$ and $g$ range over all non-decreasing surjective functions from $[0,1]$ to $[0,1]$.

Let $X$ be a set. A set $\RSpace$ where any $r \in \RSpace$ is of the form $r \subseteq X$ is called a \emph{set system} with \emph{ground set} $X$. A pair $(X,\RSpace)$ of a set system $\RSpace$ and its ground set $X$ is called a \textsc{SetCover} instance. An optimal solution to a \textsc{SetCover} instance $(X,\RSpace)$ is a set $S^*\subset\RSpace$ of minimal size, such that $\bigcup_{s\in S^*}s=X$.

\subsection{Problem definition}\label{sec:problem}
Let a polygonal curve $P$ in $\bR^d$ and a radius $\Delta>0$ together with $l\in\bN$ be given. Following Akitaya et al.~\cite{Akitaya2021Covering}, for any curve $C$ in $\bR^d$ we define
\[\Cov_P(C,\Delta) = \left(\bigcup_{0\leq s\leq t\leq 1,\,\df(P[s,t],C)\leq\Delta}[s,t]\right)\subset[0,1]\]
as the $\Delta$-\emph{coverage} of $C$, refer to Figure~\ref{fig:coverage}.
The objective is to find the smallest (w.r.t. its cardinality) set $\mathcal{C}$ of curves, each of complexity at most $l$, such that 
\[\bigcup_{C\in\mathcal{C}}\Cov_P(C,\Delta) = [0,1].\]
This can be interpreted as computing the optimal solution of the \textsc{SetCover} instance \[\left([0,1],\left\{\Cov_P(C,\Delta)\middle| C\in\bX^d_l\right\}\right).\] This point of view is the main perspective, from which we will analyze this problem.

\subsection{Related work}

The first work to appear in the line of clustering subtrajectories under the (discrete and continuous) Fréchet distance was by Buchin et al.~\cite{BuchinBGLL11}. They analyze the problem of identifying a single cluster with certain properties, such as the number of distinct subtrajectories or the length of the longest subtrajectory assigned to this cluster. They show NP-hardness-results for $(2-\eps)$-approximations as well as a matching polynomial $2$-approximation algorithm. Gudmundsson and Wong~\cite{abs-2110-15554} later presented a cubic lower-bound for the  problem of finding the largest cluster and show that this lower bound is tight. There are several more practical works that extend this approach to clustering, among them Gudmundsson and Valladares~\cite{GudmundssonV15} who presented a practical implementation on a GPU, which was later also applied to road network reconstruction from GPS data~\cite{buchinGroup2017,BuchinBGHSSSSSW20}. The algorithm repeatedly finds the largest cluster in the data and removes it, similar to the greedy \textsc{SetCover} algorithm. Buchin, Kilgus and Kölzsch~\cite{buchinGroup20} used this approach to extract migration patterns from GPS data of migrating animals.

Unfortunately, none of the clustering approaches mentioned so far offers theoretical guarantees as no explicit objective function for a cluster is formulated. 
In contrast, Agarwal et al.~\cite{agarwal2018} define an objective function for a subtrajectory clustering problem which is based on a facility location problem. They consider a weighted combination of objectives like the number of clusters, the radius of the clusters and the fraction of the curve that is not covered. For this problem, they show conditional NP-hardness results but also give a $O(\log^2 n)$-approximation algorithm in case of certain well-behaved classes of curves under the discrete Fréchet distance.

The objective function for the subtrajectory clustering problem that we want to analyze further in this paper was introduced by Akitaya et al.~\cite{Akitaya2021Covering}. They present a pseudo-polynomial bi-criterial approximation. Concretely, they introduced the problem stated in Section~\ref{sec:problem}. For given polygonal curve $P$ of complexity $n$, complexity parameter $l$ and radius $\Delta$ their algorithm finds a set of curves $\mathcal{C}$ such that $\bigcup_{C\in\mathcal{C}}\Cov_P(C,\alpha\Delta) = [0,1]$ of size $O(l^2\log(kl)k)$. Here $k$ is the smallest possible size of a set $\mathcal{C}^*$, whose coverage is $[0,1]$ and $\alpha=O(1)$. This algorithm  was subsequently improved by Brüning et al.~\cite{Brüning2022Faster} to a polynomial algorithm with an expected running time of $\Tilde{O}(k^2n + kn^3)$ and slightly better approximation bounds. Their approach is based on computing a sufficiently small \textsc{SetCover} instance with constant VC-dimension, which allows them to apply probabilistic $\eps$-net finder algorithms. The constant VC-dimension stems from considering center trajectories of complexity $l=1$ only, which is justified as any optimal subtrajectory cluster can be split at the vertices without losing coverage. This increases the number of clusters, but at the same time it decreases the VC-dimension of the set system, leading to an overall improvement in the approximation guarantees. 

\subsection{Our contribution}
The split of center curves into single edges as suggested by Brüning et al.~\cite{Brüning2022Faster}  is not particularly desirable in practice. The focus of this work is the extension to curves of non-constant complexity $l\in\bN_{\geq 2}$. While we stay in keeping with the framework suggested in~\cite{Brüning2022Faster}, the \textsc{SetCover} instances generated by our algorithm consist of centers of non-constant complexity. This comes at the expense of the running time as well as the approximation guarantee but allows the identified clusters to be of higher significance especially under practical considerations. For a given polygonal curve of complexity $n$ together with the parameter $l$ and radius $\Delta$ this new set system has size $O(ln^3)$ instead of $O(n^3)$. Each set in this set system can consist of up to $O(n)$ disjoint intervals in $[0,1]$. Storing these explicitly requires space in $O(ln^4)$ and this also bounds the total space used by our algorithm. The time to construct this set system is bounded by $O(l^2n^4)$. It further takes $O(kln^4\log(n))$ time to then greedily and deterministically find an $O(\log n)$-approximate \textsc{SetCover} solution in this set system. We observe however, that in practice this dependency is significantly better, and evaluate our approach on GPS data of ocean drifters, with a total complexity of $n\geq 10^6$. We further evaluate our approach on high-dimensional full-body motion capture data and compare the output to a state-of-the-art motion segmentation algorithm to argue the merit of our approach. 
We also demonstrate that in practice the approximation quality of our solutions is much better than suggested by the theoretical worst-case guarantees by comparing to the size of a greedily computed independent set.

\subsection{Subsequent work}

After the first publication our manuscript on arXiv, van der Hoog et al.~\cite{vanderhoog2024fasterdeterministicsubtrajectoryclustering} suggested a simple method to reduce the size of the generated set system by splitting the center curves at strategic points.  Their method leads to a candidate set of size $O(n^2 \log n)$ while increasing the number of sets in an optimal solution by a constant factor. 
We implemented and used this method in our experiments to enable the processing of larger input sets demonstrating the viability of this approach for subtrajectory clustering.

\section{A more structured set system}\label{sec:theory}

In this section we introduce the key players to our story. These allow us to impose some structure on the aforementioned set system at the cost of a constant approximation factor in both the radius, as well as the optimal solution size. With this reduction we follow Brüning et al.~\cite{Brüning2022Faster}. They combine the simplification of de Berg et al.~\cite{de2013fast} and Driemel et al.~\cite{driemelHW12} retaining the central properties of both. It is defined as follows.

\begin{definition}[$\Delta$-good simplification \cite{Brüning2022Faster}]
Let $P$ be a polygonal curve in $\bR^d$ defined by the vertices $(v_0,\ldots,v_n)$ together with a parameter $\Delta>0$ be given. We call a curve $S$ defined by the vertices $(v_{i_0},\ldots,v_{i_k})$ for $0\leq i_0<\ldots<i_k\leq n$ a simplification of $P$. Denote the polygonal curve defined by the vertices $(v_i,\ldots,v_j)$ for some $0\leq i\leq j\leq n$ by $P_{i,j}$. We say a simplification defined by the vertices $(v_{i_0},\ldots,v_{i_k})$ is $\Delta$-good if the following properties hold:
    \begin{compactenum}[(i)]
        \item $\|v_{i_j}-v_{i_{j+1}}\|\geq\frac{\Delta}{3}$ for $0\leq j<k$,
        \item $\df(P_{{i_j},{i_{j+1}}},\overline{v_{i_j}\,v_{i_{j+1}}}) \leq 3\Delta$ for all $0\leq j<k$,
        \item $\df(P_{1,{i_1}},\overline{v_{i_1}\,v_{i_1}})\leq3\Delta$,
        \item $\df(P_{{i_k},n},\overline{v_{i_k}\,v_{i_k}})\leq3\Delta$,
        \item $\df(P_{{i_j}{i_{j+2}}},\overline{v_{i_j}\,v_{i_{j+2}}}) > 2\Delta$ for all $0\leq j<k-1 $.
    \end{compactenum}
\end{definition}

These properties guarantee, that no edge is too short (\textit{(i)}), the Fréchet distance of a $\Delta$-good simplification to its underlying curve is at most $3\Delta$ (\textit{(ii) - (iv)}) and that the complexity of a $\Delta$-good simplification can not greedily be reduced (\textit{(v)}).

\begin{lemma}[\cite{Brüning2022Faster}]\label{lem:simplification}
    There is an algorithm that computes a $\Delta$-good simplification of any polygonal curve $P$ in $\bR^d$ of complexity $n$ and $\Delta>0$. Furthermore it does so in $O(n\log^2n)$ time assuming $d$ is a constant. 
\end{lemma}

Brüning et al.~\cite{Brüning2022Faster} showed that any solution of the \textsc{SetCover} instance $([0,1],\{\Cov_P(C,\Delta)\mid C\in\bX^d_l\})$ induces a solution of the \textsc{SetCover} instance with set system $\{\Cov_P(S[s,t],11\Delta)\mid 0\leq s\leq t\leq 1, |S[s,t]|\leq l\}$ where $S$ is some $\Delta$-good simplification of $P$.

\begin{theorem}[\cite{Brüning2022Faster}]\label{thm:simp}
Let a polygonal curve $P$ and values $\Delta>0$ and $l$ be given. Let $S$ be a \mbox{$\Delta$-good} simplification of $P$. Let $\mathcal{C}$ be a set of curves of size $k$ each with complexity at most $l$ such that $\bigcup_{C\in\mathcal{C}}\Cov_P(C,\Delta)=[0,1]$. Then there is a set $\mathcal{C}_S$ of size $3k$ of subcurves of $S$ each of which has complexity at most $l$ with $\bigcup_{C\in\mathcal{C}_S}\Cov_S(C,8\Delta)=[0,1]$ and thus $\bigcup_{C\in\mathcal{C}_S}\Cov_P(C,11\Delta)=[0,1]$.
\end{theorem}

\section{A finite set-system}\label{sec:finiteset}

Similar to Brüning et al.~\cite{Brüning2022Faster} we now use the set system $\{\Cov_P(S[s,t],11\Delta)\mid 0\leq s\leq t\leq 1, |S[s,t]|\leq l\}$ as an intermediary set system. Based on this intermediary set system we introduce a new set system, which consists of only $O(n^3l)$ many subcurves of some $\Delta$-good simplification $S$ of $P$.

\subsection{Extremal Candidates}

\begin{definition}[$\Delta$-free space]
    Let $P$ and $Q$ be two polygonal curves parametrized over $[0,1]$.
	The free space diagram of $P$ and $Q$ is their joint parameter space $[0,1]^2$ together with a not necessarily uniform grid, where each vertical line corresponds to a vertex of $P$ and each horizontal line to a vertex of $Q$.
	The $\Delta$-\emph{free space} of $P$ and $Q$ is defined as \[ \dfree{}{}(P,Q) = \left\{(x,y)\in[0,1]^2 \mid \|P(x) -Q(y)\|\leq\Delta \right\} \] 
	This is the set of points in the parametric space, whose corresponding points on $P$ and $Q$ are at a distance at most $\Delta$.
	The edges of $P$ and $Q$ segment the free space into cells.
	We call the intersection of $\dfree{}{}(P,Q)$ with the boundary of cells the $\Delta$-\emph{free space intervals}. Refer to Figure~\ref{fig:key} $a)$.
\end{definition}

Alt and Godau~\cite{AltG95} showed that the $\Delta$-free space inside any cell is an ellipse intersected with the cell and thus convex and of constant complexity.
They further showed that the Fréchet distance between two curves $P$ and $Q$ is less than or equal to $\Delta$ if and only if there exists a path $\pi:[0,1]\rightarrow\dfree{}{}(P,Q)$ that starts at $(0,0)$, ends in $(1,1)$ and is monotone in both coordinates. By this analysis, two free subcurves $P[a,c]$ and $Q[b,d]$ have Fréchet distance at most $\Delta$ if and only if there exists a path $\pi:[0,1]\rightarrow\dfree{}{}(P,Q)$ that starts at $(a,b)$, ends in $(c,d)$ and is monotone in both coordinates. Note that in the case of two free subcurves, $\pi$ can be monotonically increasing or monotonically decreasing in its coordinates, depending only on if $a\leq c$ or $c\leq a$ and similarly $b\leq d$ or $d\leq b$.

In order to reduce the size of the given set system for some polygonal curve $P$ and $\Delta$, we inspect the $11\Delta$-free space of a $\Delta$-good simplification $S$ of $P$ with $P$. Conceptually, we want to do the following: Start with a subcurve $S[s,t]$ of $S$ that induces the set $\Cov_P(S[s,t],11\Delta)$ in the set system. $\Cov_P(S[s,t],11\Delta)$ can by definition be described as the union of intervals $[a_i,b_i]$, such that $\df(S[s,t],P[a_i,b_i])\leq11\Delta$. Now the point $(a_i,s)$ lies in some cell, and this cell has a (not necessarily unique) left-most point. We would like to modify $s$ and $t$ as well as all $a_i$ and $b_i$ in such a way, that \textit{(i)} the new values $s'$ and $t'$ of $s$ and $t$ are defined by the $y$-coordinates of a left-/right-most point in some cell of the $11\Delta$-free space and $\textit{(ii)}$ the resulting interval $[a_i',b_i']$ includes $[a_i,b_i]$. While this is not necessarily possible, we prove that this can be achieved with a constant number of such subcurves of $S$.

\begin{figure*}
    \centering
    \includegraphics[width=\textwidth]{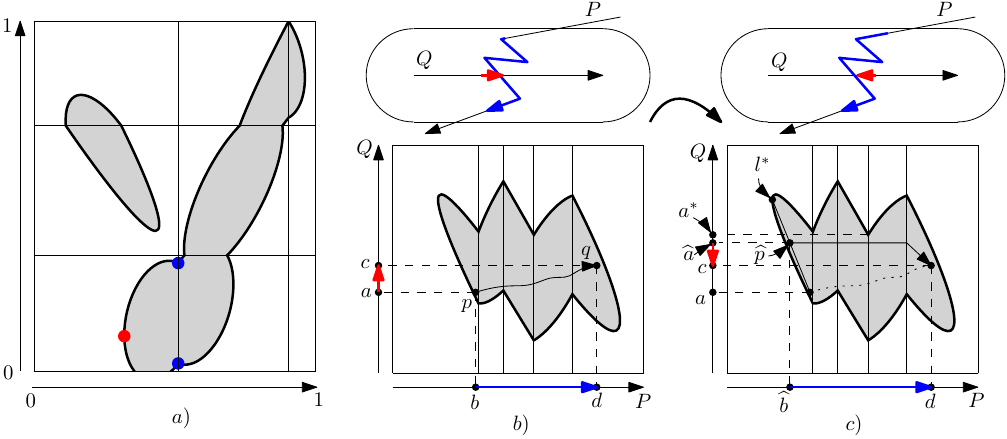}
    \caption{$a)$: Example of the $\Delta$-free space of two curves. Further illustrated is the unique left $\Delta$-extremal point in red in the free space as well as the two right $\Delta$-extremal points in blue in the lower left cell of the $\Delta$-free space. $b),c)$: Illustration to the proof of Lemma \ref{lem:key}. Depicted are the $\Delta$-free spaces of $P$ and $Q$, as well as the path from $p$ to $q$ before the modification in $b)$ and from $\widehat{p}$ to $q$ after the modification in $c)$ as in the proof of Lemma~\ref{lem:key}.}
    \label{fig:key}
\end{figure*}

\begin{definition}[$\Delta$-extremal points]
    Let $P$ and $Q$ be two polygonal curves parametrized over $[0,1]$. For every cell $C$ of the free space diagram of $P$ and $Q$ define its left $\Delta$-extremal points as either the unique left-most point inside the $\Delta$-free space in $C$ or, in case the left-most point is not unique, the upper-most and lower-most left-most point inside the $\Delta$-free space in $C$. Similarly define the right $\Delta$-extremal points of each $C$. The union of left (resp. right) $\Delta$-extremal points of all cells is defined to be the set of left (resp. right) $\Delta$-extremal points of $P$ and $Q$. Refer to Figure~\ref{fig:key} $a)$.
\end{definition}

Observe that for given polygonal curves $P$ and $Q$ of complexity $n$ the set of $\Delta$-extremal points can be computed in $O(n^2)$ time by scanning over every cell and computing these points in $O(1)$.

\begin{lemma}\label{lem:key}
Let two subcurves $P[a,c]$ and $Q[b,d]$ of polygonal curves $P$ and $Q$ as well as a value $\Delta>0$ be given such that $\df(P[a,c],Q[b,d])\leq\Delta$. Then there are values $b^*,d^*\in[0,1]$ defined by $y$-coordinates of $\Delta$-extremal points of $P$ and $Q$ such that for any $\widehat{b}$ between $b$ and $b^*$ and any $\widehat{d}$ between $d$ and $d^*$ the free subcurve $Q[\widehat{b},\widehat{d}]$ of $Q$ induces a subcurve $P[\widehat{a},\widehat{c}]$ of $P$ with $\df(P[\widehat{a},\widehat{c}],Q[\widehat{b},\widehat{d}])\leq\Delta$ and $[a,c]\subset[\widehat{a},\widehat{c}]$. 
\end{lemma}
\begin{proof}
Since $\df(P[a,c],Q[b,d])\leq\Delta$, there is a path $\pi:[0,1]\rightarrow\dfree{}{}(P,Q)$ that starts at $p=(a,b)$, ends in $q=(c,d)$ and is monotone in both coordinates. First check if the left $\Delta$-extremal point $l^*$ of the cell containing $p$ is above or below the point $p$. In case of ambiguity pick either one of the two left $\Delta$-extremal points of the cell containing $p$ to be $l^*$. If $l^*$ lies above $p$, define $b^*$ as the $y$-coordinate of the lowest left $\Delta$-extremal point in the $\Delta$-free space of $P$ and $Q$ that is above $p$. Otherwise define $b^*$ as the $y$-coordinate of the highest left $\Delta$-extremal point that is below $p$. Now let $\widehat{b}$ be given as a value inbetween $b$ and $b^*$. As $\widehat{b}$ lies in between the $y$-coordinate of $p$ and that of $l^*$, there is a point $\widehat{p}$ on the line from $p$ to $l^*$ with $y$-coordinate $\widehat{b}$. Observe that $\widehat{p}$ always lies to the left of $p$. Thus for the $x$-coordinate $\widehat{a}$ of $\widehat{p}$ we have that $[a,c]\subset[\widehat{a},c]$.

We now show that there is a monotone path from $\widehat{p}$ to $q$ inside the $\Delta$-free space and thus $\df(P[\widehat{a},c],Q[\widehat{b},d])\leq\Delta$. Assume that $\widehat{p}$ lies above $p$, as otherwise by convexity of the free space in every cell we can concatenate a straight line from $\widehat{p}$ to $p$ with $\pi$.

First, observe that we can walk straight to the right from $\widehat{p}$ until we either intersect the path $\pi$ or intersect the boundary of the cell containing $q$ on the right side. Indeed, every free space interval that contains the $y$-coordinate of $p$ also contains the $y$-coordiante of $\widehat{p}$ by definition of $b^*$. If the ray intersects $\pi$ then, again, we can construct such a path by concatenating a straight line from $\widehat{p}$ to this intersection point with the second piece of $\pi$. Hence, assume the straight line does not intersect $\pi$ until it intersects the boundary of the cell containing $q$. Then, we construct the path by walking from $\widehat{p}$ to the right, until we first enter the cell containing $q$ and then by convexity we can again connect this straight line with a second straight line to $q$ resulting in a monotone path.

Next we construct $d^*$ in a similar fashion, except we use the right-most point $r^*$ of the cell containing $q$ instead of left-most points. Yielding both $b^*$ and $d^*$ together with a path from $\widehat{p}$ to $\widehat{q}$ for any $\widehat{b}$ in between $b$ and $b^*$ and $\widehat{d}$ in between $d$ and $d^*$. 
\end{proof}


\begin{figure*}
\centering
    \includegraphics[width=\textwidth]{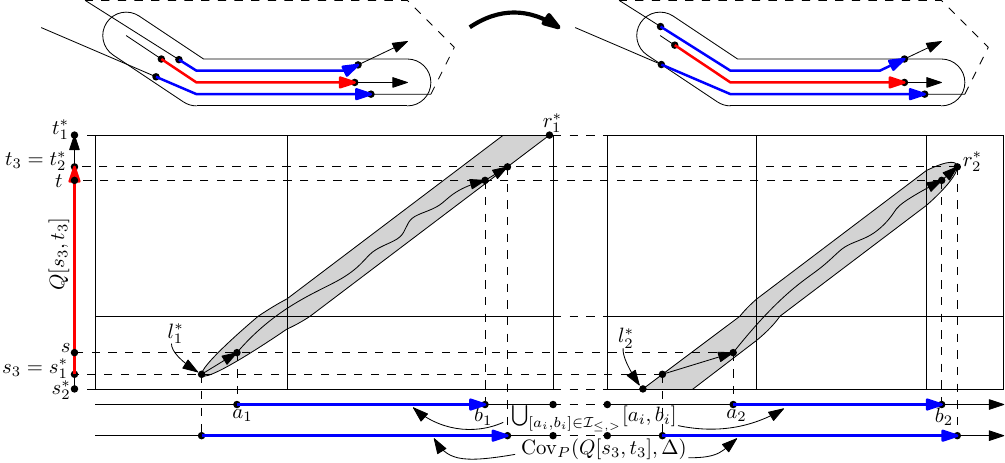}
    \vspace{-0.3cm}
    \caption{Illustration for the proof of Theorem~\ref{thm:main1}. The inclusion-wise increase of the $\Delta$-Coverage for all intervals in $\mathcal{I}_{\leq,>}=\{[a_1,b_1],[a_2,b_2]\}$ is depicted. The intervals $[a_1,b_1]$ and $[a_2,b_2]$ are in $\mathcal{I}_{\leq,>}$ because the left $\Delta$-extremal points $l_1^*$ and $l_2^*$ of the cells containing $(a_1,s)$ and $(a_2,s)$ lie below $s$, and similarly $r_1^*$ and $r_2^*$ lie above $t$.}
    \label{fig:thmkey}
\end{figure*}

\begin{theorem}\label{thm:main1}
Let $P$ and $Q$ be two given polygonal curves together with a value $\Delta>0$.
For every $s,t\in[0,1]$ there are at most eight values $s_1,\ldots,s_4,t_1,\ldots,t_4\in[0,1]$ defined by $y$-coordinates of $\Delta$-extremal points of $P$ and $Q$ such that
\[\Cov_P(Q[s,t],\Delta)\subset\bigcup_{i=1}^4\Cov_P(Q[s_i,t_i],\Delta).\]
\end{theorem}
\begin{proof}
An illustration to this proof can be seen in Figure~\ref{fig:thmkey}.
Let $\mathcal{I}=\{[a_i,b_i]\subset[0,1]\mid\df(P[a_i,b_i],Q[s,t])\leq\Delta\}$ be the set of intervals defining $\Cov_P(Q[s,t],\Delta)$. Partition $\mathcal{I}$ into four sets, according to whether for $[a_i,b_i]$ the values $s_i^*$ and $t_i^*$ from Lemma \ref{lem:key} is above or below $s$ and above or below $t$. That is
\begin{multicols}{2}
\begin{align*}
    \mathcal{I}_{\leq,\leq}&=\{[a_i,b_i]\in\mathcal{I}\mid s_i^*\leq s,t_i^*\leq t\},\\
    \mathcal{I}_{>,\leq}&=\{[a_i,b_i]\in\mathcal{I}\mid s_i^*> s,t_i^*\leq t\},
\end{align*}

\begin{align*}
    \mathcal{I}_{\leq,>}&=\{[a_i,b_i]\in\mathcal{I}\mid s_i^*\leq s,t_i^*> t\},\\
    \mathcal{I}_{>,>}&=\{[a_i,b_i]\in\mathcal{I}\mid s_i^*> s,t_i^*> t\}.
\end{align*}
\end{multicols}

Now for $\mathcal{I}_{\leq,\leq}$ observe that the set of $s_i^*$ as well as the set of $t_i^*$ is finite, as there can be at most $O(n^2)$ many distinct left or right $\Delta$-extremal points. Thus 
\begin{multicols}{2}
\begin{align*}
s_1&=\max\{s_i^*\mid [a_i,b_i]\in\mathcal{I}_{\leq,\leq}\},
\end{align*}

\begin{align*}
t_1&=\max\{t_i^*\mid [a_i,b_i]\in\mathcal{I}_{\leq,\leq}\}.
\end{align*}
\end{multicols}
are well defined, and $s_1$ lies between $s$ and any $s_i^*$ in $\{s_i^*\mid [a_i,b_i]\in\mathcal{I}_{\leq,\leq}\}$ and $t_1$ lies between $t$ and any $t_i^*$ in $\{t_i^*\mid [a_i,b_i]\in\mathcal{I}_{\leq,\leq}\}$. Thus by Lemma \ref{lem:key} There are values $\widehat{a_i}$ and $\widehat{b_i}$ for every $[a_i,b_i]\in\mathcal{I}_{\leq,\leq}$, such that $\df(P[\widehat{a_i},\widehat{b_i}],Q[s_1,t_1])\leq\Delta$, and thus 
\[\bigcup_{[a_i,b_i]\in\mathcal{I}_{\leq,\leq}}[a_i,b_i]\subset\Cov_Q(P[s_1,t_1],\Delta).\]
Similarly we can identify
\begin{multicols}{2}
\begin{align*}
s_2&=\min\{s_i^*\mid [a_i,b_i]\in\mathcal{I}_{>,\leq}\},\\
s_3&=\max\{s_i^*\mid [a_i,b_i]\in\mathcal{I}_{\leq,>}\},\\
s_4&=\min\{s_i^*\mid [a_i,b_i]\in\mathcal{I}_{>,>}\},
\end{align*}%

\begin{align*}
t_2&=\max\{t_i^*\mid [a_i,b_i]\in\mathcal{I}_{>,\leq}\},\\
t_3&=\min\{t_i^*\mid [a_i,b_i]\in\mathcal{I}_{\leq,>}\},\\
t_4&=\min\{t_i^*\mid [a_i,b_i]\in\mathcal{I}_{>,>}\}.
\end{align*}
\end{multicols}
As $\mathcal{I}_{\leq,\leq}$, $\mathcal{I}_{>,\leq}$,$\mathcal{I}_{\leq,>}$ and $\mathcal{I}_{>,>}$ partition $\mathcal{I}$, it follows that
\begin{align*}
    \Cov_P(Q[s,t],\Delta)&=\bigcup_{[a_i,b_i]\in\mathcal{I}}[a_i,b_i]\\
    &\subset\bigcup_{i=1}^4\Cov_P(Q[s_i,t_i],\Delta),
\end{align*}
proving the claim.
\end{proof}

\begin{theorem}\label{thm:main2}
Let a polygonal curve $P$ in $\bR^d$ of complexity $n$ and values $\Delta>0$ and $l\in\bN$ be given. A set of curves $\mathcal{S}\subset\bX_l^d$ of size $O(n^3l)$ can be computed in time $O(n^3l^2)$ with the following property. If there is a set $\mathcal{C}\subset\bX_l^d$ of size $k$ such that $\bigcup_{C\in\mathcal{C}}C=[0,1]$, then there is a set $\mathcal{C}_\mathcal{S}\subset \mathcal{S}$ of size $12k$ with $\bigcup_{C\in\mathcal{C}_\mathcal{S}}C=[0,1]$.
\end{theorem}
\begin{proof}
    First compute a $\Delta$-good simplification $S$ of $P$ in time $O(n\log^2 n)$, which is possible by Lemma \ref{lem:simplification}. Next compute the $11\Delta$-free space and subsequently the left and right $11\Delta$-extremal points of $S$ and $P$ in time $O(n^2)$. For every left $11\Delta$-extremal point $s$ and every right $11\Delta$-extremal point $t$ such that $|S[s,t]|\leq l$ store $S[s,t]$ in $\mathcal{S}$. This takes $O(n^2)\cdot O(ln)\cdot O(l)$ time.
    The correctness of the property is an immediate consequence of Theorem \ref{thm:simp} and Theorem \ref{thm:main1}.
\end{proof}

\subsection{Discretization of the ground set}\label{sec:discretization}

So far we have presented a discretization of the set system to our original \textsc{SetCover} instance, which preserves the optimal solution up to a bounded approximation factor. This allows us to apply \textsc{SetCover} solving techniques such as a greedy algorithm \cite{slavik1996tight}, which chooses sets incrementally and in each step chooses the set which increases the accumulated coverage of the ground set the most. In the following, we argue that this can be done efficiently.

\begin{lemma}\label{lem:discrete}
    Let $P$ and $Q$ be two polygonal curves of complexity at most $n$ and a parameter $\Delta>0$ be given. Then $\Cov_P(Q,\Delta)$ can be written as the union of $O(n)$ disjoint closed intervals.
\end{lemma}
\begin{proof}
    Let $P[a,b]$ be a subcurve of $P$, such that $\df(P[a,b],Q)\leq\Delta$. Then, $(a,0)$ is contained in the $\Delta$-free space of some cell. Let $(a^*,0)$ be the leftmost point in the cell containing $(a,0)$. We can concatenate $\overline{a^*\,a}$ to any monotone path starting in $a$. Thus, any path starting in the cell containing $(a,0)$ may as well start in $(a^*,0)$. Similarly, any path may as well end in a rightmost point $(b^*,1)$ of some cell. Thus, $\Cov_Q(P,\Delta)$ consists of a union of intervals,  each interval starting and ending in one of at most $n$ points.
\end{proof}

\begin{theorem}
    Let $P$ be a polygonal curve, and let $\Delta>0$ and $l\in\bN$ be given. There is an algorithm, which computes a $O(\log(n))$-approximation for the \textsc{SetCover} instance \mbox{$([0,1],\{\Cov_P(S,11\Delta)\mid S\in\mathcal{S}\})$}, where $\mathcal{S}$ is as in Theorem \ref{thm:main2}. 
\end{theorem}
\begin{proof}
An immediate consequence of Lemma \ref{lem:discrete} is that the arrangement of all sets contained in some set of the set family consists of at most $O(n^4l)$ intervals $\mathcal{I}$. Interpreting these intervals as discrete objects, we now get a transformed \textsc{SetCover} instance \[\left(\mathcal{I},\left\{\bigcup_{I\in\mathcal{I},I\subset\Cov_P(S,11\Delta)}I\middle| S\in\mathcal{S}\right\}\right).\]
Applying the greedy \textsc{SetCover} algorithm to this yields a $O(\log(n^4l))=O(\log n)$ approximation algorithm in the size of an optimal solution.
\end{proof}

\begin{corollary}
    Let a polygonal curve $P$ of complexity $n$ and values $\Delta>0$ and $l\in\bN$ be given. There exists an algorithm that computes a set of curves $\mathcal{C}$, each of complexity at most $l$ such that $\bigcup_{C\in\mathcal{C}}\Cov_P(C,11\Delta)=[0,1]$ in $\Tilde{O}(l^2n^4 + kln^4)$ time. Further $|\mathcal{C}|\in O(\log(n)k)$, where $k$ is the smallest cardinality of any set $\mathcal{C}^*\subset\bX^d_l$ such that $\bigcup_{C\in\mathcal{C}}\Cov_P(C,11\Delta)$.
\end{corollary}
\begin{proof}
    It is left to show that for any subcurve $C$ of $\mathcal{S}$ of complexity $l$ we can compute $\Cov_{\mathcal{P}}(C,11\Delta)$ in $O(ln)$ time given the $11\Delta$-Freespace of $\mathcal{S}$ and $P$. Further we need to show that for a given partial solution $S^*$ of size $O(\log (n)k)$ and a given subcurve $C$ of $\mathcal{S}$ we can compute the covered fraction from the intervals of the arrangement in $\Cov_P(C,11\Delta)\setminus\bigcup_{S\in S^*}\Cov_P(S,11\Delta)$ to update $C$ for the greedy \textsc{SetCover} algorithm in $\Tilde{O}(n)$.\\
    The first part is straightforward, as the part of the $11\Delta$-Freespace we have to consider is of size $O(ln)$. We need to traverse it once from bottom left to top right to compute the $11\Delta$-Coverage of such a subcurve. For the second part, observe that $k\leq n$, as a possible solution consists of all the edges of $P$. Thus, as any $11\Delta$-coverage consists of at most $O(n)$ intervals the $11\Delta$-Coverage of any partial solution consists of at most $O(n^2)$ disjoint intervals, which we store in an interval-tree. Any interval in this interval-tree also stores the number of intervals of the arrangement that are not covered and lie to the left of its left boundary. With this information we can compute the number of intervals of the arrangement in $\Cov_P(C,11\Delta)\setminus\bigcup_{S\in S^*}\Cov_P(S,11\Delta)$ in $O(n\log(n^2))=O(n\log n)$ time. Updating the information in the interval-tree upon increasing the partial solution can be done in $\Tilde{O}(n^3)$ time and needs to be done at most $O(k\log n)$ times. 
\end{proof}

\section{Experiments}

All experiments were conducted on a Linux system with 16GB of memory with an Intel \texttt{i5-9600} CPU, a decent CPU with $6$ cores, but far from the fastest hardware available.

\subsection{Implementation Details}
Subsequent to the first publication of our results on arxiv van der Hoog et al.~\cite{vanderhoog2024fasterdeterministicsubtrajectoryclustering} observed that it suffices to consider candidates defined by pairs $(a,b)$ of extremal coordinates such that $2^i$ other extremal coordinates lie in $[a,b]$ for some $i\in\bN$. This reduces the cardinality of the candidate set to $O(n^2\log n)$ transforming each center curve in any solution into at most two new center curves of similar complexity. This candidate set defines the set cover instance which is constructed and then solved in our implementation. The approximation guarantees of $11\Delta$ in the distance threshold and $O(k\log(n))$ in the solution size stay unchanged, while the theoretical running time improves to $\tilde{O}(kln^3)$.

Our \texttt{C++}-implementation can be found at~\cite{anonymousGithub}. The algorithm is given a set $\{P_1,\ldots,P_m\}$ of curves and three parameters $\Delta_\textit{simp}$ and $\Delta_\textit{free}$ and $l$. It first computes simplifications $S_i$ of $P_i$ for all $i$ with parameter $\Delta_\textit{simp}$, and then computes the $\Delta_\textit{free}$-free space of $S_i$ and $S_j$ as well as their extremal points for all pairs $(i,j)\in\{1,\ldots,m\}^2$. Next it computes all pairs of extremal coordinates that $(i)$ lie on the same curve, $(ii)$ the resulting subcurve has complexity at most $l$ and $(iii)$ there is a power of $2$ of other extremal coordinates in the interval between them. For every such pair of extremal coordinates $(a,b)$ on curve $S_i$ we compute the $\Delta_\textit{free}$-coverage of $S_i[a,b]$ via the computed $\Delta_\textit{free}$-spaces of $S_i$ with any other $S_j$. This $\Delta_\textit{free}$-coverage is a subset of $\{1,\ldots,m\}\times[0,1]$ and defines the set cover instance. This set cover instance we solve via a greedy set cover algorithm that iteratively picks and adds the candidate which maximizes the arc-length (computed on each $S_i$) of the additional coverage. This step is repeated until $\{1,\ldots,m\}\times[0,1]$ is covered. 

We recursively pick the candidate maximizing the arc-length---instead of the number of intervals of the induced arrangement of coverages---of the added $\Delta$-coverage. This follows the $\Delta$-coverage maximization discussion in \cite{Brüning2022Faster} and allows us to stop the greedy algorithm after a small number of rounds and still have a partial solution that covers a large fraction of the input. Further, we introduced the two parameters $\Delta_{\textit{simp}}$ and $\Delta_{\textit{free}}$ to test the stability of the threshold parameter $\Delta$ in both the simplification and free-space computation step. Given some $\Delta$, setting $\Delta_{\textit{simp}}=3\Delta$ and $\Delta_{\textit{free}}=8\Delta$ reflects the theoretical results. 

\subsection{Ocean Drifters}
\begin{figure}
    \centering	
    \subfloat[\centering varying $\Delta_\textit{simp}$ and $\Delta_\textit{free}$ with $3:8$ ratio.\label{subfig:1}]{\scalebox{0.5}{\begin{tikzpicture}
\begin{axis}[
    xlabel={Points},
    ylabel={Time in $[s]$},
    legend pos=north west,
    grid=major,
    width=\textwidth,
    ymode=log,
    xmode=log,
    height=0.6\textwidth,
    style={ultra thick},
]
\addplot coordinates {
    (10000.0, 0.27719900989905)
    (20000.0, 0.6316489949822426)
    (30000.0, 1.188493647612631)
    (40000.0, 1.197291468270123)
    (50000.0, 1.2472961102612317)
    (60000.0, 1.5211414876393974)
    (70000.0, 1.621997483074665)
    (80000.0, 1.887173877563328)
    (90000.0, 2.11663400195539)
    (100000.0, 2.303281505592168)
    (200000.0, 5.272964080795646)
    (300000.0, 8.847055089194328)
    (400000.0, 13.526936966460198)
    (500000.0, 18.903752747923136)
    (600000.0, 26.31807895982638)
    (700000.0, 34.53292950661853)
    (800000.0, 43.89596584625542)
    (900000.0, 53.70135052455589)
    (1000000.0, 64.97437348077074)
    (2000000.0, 432.7205221820623)
};
\addlegendentry{$\Delta_{simp} = 30000, \Delta_{free} = 80000, l = 3$}

\addplot coordinates {
    (10000.0, 0.445091408008011)
    (20000.0, 1.423659845997463)
    (30000.0, 1.826881044995389)
    (40000.0, 1.8933151230157816)
    (50000.0, 2.1811309380136663)
    (60000.0, 2.834360322987777)
    (70000.0, 2.912917552006548)
    (80000.0, 3.283769465997466)
    (90000.0, 3.694548112995108)
    (100000.0, 3.945654797993484)
    (200000.0, 8.230075030005537)
    (300000.0, 12.39505133999046)
    (400000.0, 17.151918208997813)
    (500000.0, 22.09512110601645)
    (600000.0, 28.270884515994112)
    (700000.0, 33.25838872800523)
    (800000.0, 39.72566475901112)
    (900000.0, 47.33260071399855)
    (1000000.0, 55.22182434500428)
    (2000000.0, 180.5534527479904)
};
\addlegendentry{$\Delta_{simp} = 60000, \Delta_{free} = 160000, l = 3$}

\addplot coordinates {
    (10000.0, 0.8452739970089169)
    (20000.0, 2.2322320279781707)
    (30000.0, 2.3103071320074378)
    (40000.0, 2.6888219279935583)
    (50000.0, 3.369186126001296)
    (60000.0, 4.266016601002775)
    (70000.0, 4.657629793015076)
    (80000.0, 5.036690647000796)
    (90000.0, 5.711281853000401)
    (100000.0, 5.991882975999033)
    (200000.0, 12.578254264008136)
    (300000.0, 18.82653490299708)
    (400000.0, 24.53256264899392)
    (500000.0, 30.581985822005663)
    (600000.0, 37.93250536100823)
    (700000.0, 44.36676167300902)
    (800000.0, 51.76760081699467)
    (900000.0, 60.47053217800567)
    (1000000.0, 68.33913648399175)
    (2000000.0, 164.46900820099108)
};
\addlegendentry{$\Delta_{simp} = 90000, \Delta_{free} = 240000, l = 3$}

\addplot coordinates {
    (10000.0, 1.12916148500517)
    (20000.0, 2.5897573309921427)
    (30000.0, 2.7477634560054867)
    (40000.0, 3.2908069170080125)
    (50000.0, 4.658696018988849)
    (60000.0, 5.209048511009314)
    (70000.0, 5.751090222009225)
    (80000.0, 6.639542246994097)
    (90000.0, 7.905111605010461)
    (100000.0, 7.985967337997863)
    (200000.0, 16.268287960992893)
    (300000.0, 23.82462491300248)
    (400000.0, 31.20564111499698)
    (500000.0, 39.03850103600416)
    (600000.0, 47.97652842701064)
    (700000.0, 55.645574028021656)
    (800000.0, 65.55798476099153)
    (900000.0, 73.57997708501352)
    (1000000.0, 82.55346960401221)
    (2000000.0, 185.80304687400348)
};
\addlegendentry{$\Delta_{simp} = 120000, \Delta_{free} = 320000, l = 3$}

\addplot[
    domain=10000:2000000, 
    samples=100,     
    dashed,          
] {0.000005*x^1.2};
\addlegendentry{$n^{1.2}$}

\end{axis}
\end{tikzpicture}

}}
    \subfloat[\centering varying $l$.\label{subfig:2}]{\scalebox{0.5}{\begin{tikzpicture}
\begin{axis}[
    xlabel={Points},
    ylabel={Time in $[s]$},
    legend pos=north west,
    grid=major,
    width=\textwidth,
    ymode=log,
    xmode=log,
    height=0.6\textwidth,
    style={ultra thick},
]
\addplot coordinates {
    (10000.0, 0.4407824929949129)
    (20000.0, 1.4143317289999686)
    (30000.0, 1.8724696459976256)
    (40000.0, 1.9828219079936389)
    (50000.0, 2.2229771920101484)
    (60000.0, 3.027941534019192)
    (70000.0, 3.041258722994826)
    (80000.0, 3.212228035001317)
    (90000.0, 3.853493583999807)
    (100000.0, 4.067440043989336)
    (200000.0, 8.897705252995365)
    (300000.0, 13.166170385011355)
    (400000.0, 18.236916852998547)
    (500000.0, 23.386700301998644)
    (600000.0, 29.549029109999537)
    (700000.0, 35.029031757992925)
    (800000.0, 41.9002918160113)
    (900000.0, 49.63445252101519)
    (1000000.0, 57.733437086993945)
    (2000000.0, 180.43967428200995)
};
\addlegendentry{$\Delta_{simp} = 60000, \Delta_{free} = 160000, l = 1$}

\addplot coordinates {
    (10000.0, 0.445091408008011)
    (20000.0, 1.423659845997463)
    (30000.0, 1.826881044995389)
    (40000.0, 1.8933151230157816)
    (50000.0, 2.1811309380136663)
    (60000.0, 2.834360322987777)
    (70000.0, 2.912917552006548)
    (80000.0, 3.283769465997466)
    (90000.0, 3.694548112995108)
    (100000.0, 3.945654797993484)
    (200000.0, 8.230075030005537)
    (300000.0, 12.39505133999046)
    (400000.0, 17.151918208997813)
    (500000.0, 22.09512110601645)
    (600000.0, 28.270884515994112)
    (700000.0, 33.25838872800523)
    (800000.0, 39.72566475901112)
    (900000.0, 47.33260071399855)
    (1000000.0, 55.22182434500428)
    (2000000.0, 180.5534527479904)
};
\addlegendentry{$\Delta_{simp} = 60000, \Delta_{free} = 160000, l = 3$}

\addplot coordinates {
    (10000.0, 0.437488817013218)
    (20000.0, 1.4072223669936648)
    (30000.0, 1.854022226005327)
    (40000.0, 1.9000229980010768)
    (50000.0, 2.1743654260062613)
    (60000.0, 3.026665106997825)
    (70000.0, 3.082566032986506)
    (80000.0, 3.1590875609981595)
    (90000.0, 3.755504553992069)
    (100000.0, 3.8862941330007743)
    (200000.0, 8.161216189997504)
    (300000.0, 12.387453304996598)
    (400000.0, 16.966822016998776)
    (500000.0, 21.885099981998792)
    (600000.0, 27.48148662901076)
    (700000.0, 32.59366995001619)
    (800000.0, 38.65000529203098)
    (900000.0, 45.80716944599408)
    (1000000.0, 52.398457770992536)
    (2000000.0, 167.0829270760005)
};
\addlegendentry{$\Delta_{simp} = 60000, \Delta_{free} = 160000, l = 5$}

\addplot coordinates {
    (10000.0, 0.4379752590029966)
    (20000.0, 1.4188520530005917)
    (30000.0, 1.866910570999608)
    (40000.0, 1.8840835459996013)
    (50000.0, 2.183591212000465)
    (60000.0, 2.9813540700124577)
    (70000.0, 2.871423086995492)
    (80000.0, 3.222934413031908)
    (90000.0, 3.5378765469940845)
    (100000.0, 3.884478481981205)
    (200000.0, 8.32312579403515)
    (300000.0, 12.36004210700048)
    (10000.0, 0.4446657669905107)
    (20000.0, 1.4146483310032636)
    (30000.0, 1.8607265850005208)
    (40000.0, 1.89114908198826)
    (50000.0, 2.1550448649795726)
    (60000.0, 2.823271247034427)
    (70000.0, 2.883416313008638)
    (80000.0, 3.0729179280169774)
    (90000.0, 3.524631918000523)
    (100000.0, 3.66546641101013)
    (200000.0, 7.936748792999424)
    (300000.0, 11.775325979979243)
    (400000.0, 16.07262964299298)
    (500000.0, 20.96899864802253)
    (600000.0, 26.264883536001435)
    (700000.0, 31.13123064200045)
    (800000.0, 36.73213443800341)
    (900000.0, 42.756052195996745)
    (1000000.0, 48.56842555501498)
    (2000000.0, 155.56586735098972)
};
\addlegendentry{$\Delta_{simp} = 60000, \Delta_{free} = 160000, l = 10$}

\addplot[
    domain=10000:2000000, 
    samples=100,     
    dashed,          
] {0.000005*x^1.2};
\addlegendentry{$n^{1.2}$}

\end{axis}
\end{tikzpicture}

}}\\
    \vspace{3mm}
    \subfloat[\centering varying $\Delta_\textit{simp}$.\label{subfig:3}]{\scalebox{0.5}{

\begin{tikzpicture}
\begin{axis}[
    xlabel={Points},
    ylabel={Time in $[s]$},
    legend pos=north west,
    grid=major,
    width=\textwidth,
    ymode=log,
    xmode=log,
    height=0.6\textwidth,
    ymax=1000,
    style={ultra thick},
]
\addplot coordinates {
    (10000.0, 0.2743319664150476)
    (20000.0, 0.6291411239653826)
    (30000.0, 1.156669078860432)
    (40000.0, 1.2013287860900164)
    (50000.0, 1.2719988357275724)
    (60000.0, 1.5111245377920568)
    (70000.0, 1.6127860150299966)
    (80000.0, 1.864451705943793)
    (90000.0, 2.0716149527579546)
    (100000.0, 2.2585983057506382)
    (200000.0, 5.118592428974807)
    (300000.0, 8.476169132161885)
    (400000.0, 12.684599642641842)
    (500000.0, 17.49183826101944)
    (600000.0, 23.824896934442226)
    (700000.0, 31.145028330385685)
    (800000.0, 40.83896321756765)
    (900000.0, 52.7897315020673)
    (1000000.0, 64.28902415791526)
    (2000000.0, 325.2638315558433)
};
\addlegendentry{$\Delta_{simp} = 30000, \Delta_{free} = 160000, l = 3$}

\addplot coordinates {
    (10000.0, 0.445091408008011)
    (20000.0, 1.423659845997463)
    (30000.0, 1.826881044995389)
    (40000.0, 1.8933151230157816)
    (50000.0, 2.1811309380136663)
    (60000.0, 2.834360322987777)
    (70000.0, 2.912917552006548)
    (80000.0, 3.283769465997466)
    (90000.0, 3.694548112995108)
    (100000.0, 3.945654797993484)
    (200000.0, 8.230075030005537)
    (300000.0, 12.39505133999046)
    (400000.0, 17.151918208997813)
    (500000.0, 22.09512110601645)
    (600000.0, 28.270884515994112)
    (700000.0, 33.25838872800523)
    (800000.0, 39.72566475901112)
    (900000.0, 47.33260071399855)
    (1000000.0, 55.22182434500428)
    (2000000.0, 180.5534527479904)
};
\addlegendentry{$\Delta_{simp} = 60000, \Delta_{free} = 160000, l = 3$}

\addplot coordinates {
    (10000.0, 0.8611364999960642)
    (20000.0, 2.226682918990264)
    (30000.0, 2.3011184420029167)
    (40000.0, 2.670405172000756)
    (50000.0, 3.355560283001978)
    (60000.0, 4.254892120996374)
    (70000.0, 4.555174457011162)
    (80000.0, 4.85947609100549)
    (90000.0, 5.67829608800821)
    (100000.0, 5.924230700999033)
    (200000.0, 12.230114259989932)
    (300000.0, 18.662654179992387)
    (400000.0, 24.70066890500312)
    (500000.0, 30.946579371011467)
    (600000.0, 38.05862567599979)
    (700000.0, 44.885458635006216)
    (800000.0, 52.30285250899033)
    (900000.0, 60.72542601398891)
    (1000000.0, 68.88472725599422)
    (2000000.0, 166.99517136599752)
};
\addlegendentry{$\Delta_{simp} = 90000, \Delta_{free} = 160000, l = 3$}

\addplot coordinates {
    (10000.0, 1.1267777110042516)
    (20000.0, 2.5595191930042347)
    (30000.0, 2.6602265930123394)
    (40000.0, 3.312578848999692)
    (50000.0, 4.807987031002995)
    (60000.0, 5.5497717749967705)
    (70000.0, 5.704909836989828)
    (80000.0, 6.6083954379864736)
    (90000.0, 7.484182188010891)
    (100000.0, 8.145490499999141)
    (200000.0, 16.287324280987377)
    (300000.0, 23.940453852992505)
    (400000.0, 31.20774520999112)
    (500000.0, 39.06129344800138)
    (600000.0, 47.698589658000856)
    (700000.0, 56.44028363001416)
    (800000.0, 65.76537355101027)
    (900000.0, 74.50175737100653)
    (1000000.0, 83.44854148602462)
    (2000000.0, 190.71854418901785)
};
\addlegendentry{$\Delta_{simp} = 120000, \Delta_{free} = 160000, l = 3$}

\addplot[
    domain=10000:2000000, 
    samples=100,     
    dashed,          
] {0.000005*x^1.2};
\addlegendentry{$n^{1.2}$}

\end{axis}
\end{tikzpicture}
}}
    \subfloat[\centering varying $\Delta_\textit{free}$.\label{subfig:4}]{\scalebox{0.5}{\begin{tikzpicture}
\begin{axis}[
    xlabel={Points},
    ylabel={Time in $[s]$},
    legend pos=north west,
    grid=major,    
    width=\textwidth,
    ymode=log,
    xmode=log,
    height=0.6\textwidth,
    style={ultra thick},
]
\addplot coordinates {
    (10000.0, 0.27719900989905)
    (20000.0, 0.6316489949822426)
    (30000.0, 1.188493647612631)
    (40000.0, 1.197291468270123)
    (50000.0, 1.2472961102612317)
    (60000.0, 1.5211414876393974)
    (70000.0, 1.621997483074665)
    (80000.0, 1.887173877563328)
    (90000.0, 2.11663400195539)
    (100000.0, 2.303281505592168)
    (200000.0, 5.272964080795646)
    (300000.0, 8.847055089194328)
    (400000.0, 13.526936966460198)
    (500000.0, 18.903752747923136)
    (600000.0, 26.31807895982638)
    (700000.0, 34.53292950661853)
    (800000.0, 43.89596584625542)
    (900000.0, 53.70135052455589)
    (1000000.0, 64.97437348077074)
    (2000000.0, 432.7205221820623)
};
\addlegendentry{$\Delta_{simp} = 30000, \Delta_{free} = 80000, l = 3$}

\addplot coordinates {
    (10000.0, 0.2743319664150476)
    (20000.0, 0.6291411239653826)
    (30000.0, 1.156669078860432)
    (40000.0, 1.2013287860900164)
    (50000.0, 1.2719988357275724)
    (60000.0, 1.5111245377920568)
    (70000.0, 1.6127860150299966)
    (80000.0, 1.864451705943793)
    (90000.0, 2.0716149527579546)
    (100000.0, 2.2585983057506382)
    (200000.0, 5.118592428974807)
    (300000.0, 8.476169132161885)
    (400000.0, 12.684599642641842)
    (500000.0, 17.49183826101944)
    (600000.0, 23.824896934442226)
    (700000.0, 31.145028330385685)
    (800000.0, 40.83896321756765)
    (900000.0, 52.7897315020673)
    (1000000.0, 64.28902415791526)
    (2000000.0, 325.2638315558433)
};
\addlegendentry{$\Delta_{simp} = 30000, \Delta_{free} = 160000, l = 3$}

\addplot coordinates {
    (10000.0, 0.2721338509581983)
    (20000.0, 0.6271378989331424)
    (30000.0, 1.1511705508455634)
    (40000.0, 1.2129436950199306)
    (50000.0, 1.2628406109288337)
    (60000.0, 1.497494414448738)
    (70000.0, 1.6052325130440297)
    (80000.0, 1.873836417216808)
    (90000.0, 2.059698451776057)
    (100000.0, 2.241993742063641)
    (200000.0, 4.945766382850707)
    (300000.0, 8.18013993324712)
    (400000.0, 12.20221052877605)
    (500000.0, 16.729992508888245)
    (10000.0, 0.2775487289763987)
    (20000.0, 0.6293183108791709)
    (30000.0, 1.1546574020758271)
    (40000.0, 1.206835452001542)
    (50000.0, 1.260572717525065)
    (60000.0, 1.509505104739219)
    (70000.0, 1.5941020660102367)
    (80000.0, 1.8642137479037049)
    (90000.0, 2.0735959513112903)
    (100000.0, 2.2317455909214914)
    (200000.0, 4.945315834600478)
    (300000.0, 7.982077595312148)
    (400000.0, 11.868728865869343)
    (500000.0, 16.231448793783784)
    (600000.0, 21.8267435580492)
    (700000.0, 29.92411387711764)
    (800000.0, 38.9079528618604)
    (900000.0, 48.131606252864)
    (1000000.0, 57.04510872485116)
    (2000000.0, 257.08957011718303)
};
\addlegendentry{$\Delta_{simp} = 30000, \Delta_{free} = 240000, l = 3$}

\addplot coordinates {
    (10000.0, 0.2717683410010068)
    (20000.0, 0.6288525609998032)
    (30000.0, 1.1700122149923118)
    (40000.0, 1.2124618840025505)
    (50000.0, 1.259241852996638)
    (60000.0, 1.5159145620127674)
    (70000.0, 1.6005109199904837)
    (80000.0, 1.9243741660029627)
    (90000.0, 2.0987833160033915)
    (100000.0, 2.339744126002188)
    (200000.0, 5.000183508993359)
    (300000.0, 8.264458531004493)
    (400000.0, 12.257520947998271)
    (500000.0, 16.647571672001504)
    (600000.0, 22.71155708801234)
    (700000.0, 29.9072454309935)
    (800000.0, 38.79661530601152)
    (900000.0, 48.92615313499118)
    (1000000.0, 57.31145696398744)
    (2000000.0, 229.38438491499983)
};
\addlegendentry{$\Delta_{simp} = 30000, \Delta_{free} = 320000, l = 3$}

\addplot[
    domain=10000:2000000, 
    samples=100,     
    dashed,          
] {0.000005*x^1.2};
\addlegendentry{$n^{1.2}$}

\end{axis}
\end{tikzpicture}

}}
    \caption{Influence of different combinations of parameters on the running time evaluated on the data set from the NOAA Global Drifter Program~\cite{gdacDataset}.}
    \label{fig:drifters}
\end{figure}

We apply our algorithm to trajectories from the NOAA Global Drifter Program~\cite{gdacDataset}. This is a comprehensive data set consisting of almost $20\,000$ ocean surface drifters that have been released across the ocean as far back as 1979.
For the evaluation, we focus on the subset of trajectories consisting of all drifters recorded in the last year (2022 - 2024). This data set consist of $5500$ different trajectories which consist on average of $500$ points resulting in a total input complexity of $n\geq 2\cdot10^6$. Refer to Figure~\ref{fig:alldrifters} in which the data set and the computed clustering with $\Delta_\textit{simp}=20\km$, $\Delta_\textit{free}=400\km$ and $l=20$ is depicted. 

\subparagraph{Evaluation} We apply our techniques with a range of radii with $\Delta_{\textit{simp}}$ between $5\km$ and $120\km$, $\Delta_{\textit{free}}$ between $10\km$ and $320\km$, as well as a range of complexity bounds with $l$ between $1$ and $10$. Figure \ref{fig:drifters} shows the running times. We observe that the running time appears to be mostly independent of the exact values of $\Delta_\textit{simp}$ and $\Delta_\textit{free}$, and scales favorably in $n$ compared to the theoretical results. In the $n\leq10^5$ regime, it appears to scale near-linear with an observed running time of roughly $O(n^{1.2})$. With increasing $n$ it approaches roughly quadratic complexity ($O(n^2)$) compared to the theoretical running time of $O(kn^3\log^3n)$. 
In addition, we empirically evaluate the approximation ratio of our set cover algorithm using the size of a greedily computed independent set as a lower bound. We observe that for all tested instances the approximation ratio is less than $3$.

\subsection{Full-Body Motion Tracking Data}
\begin{figure}
    \centering
    \includegraphics[width=0.9\textwidth]{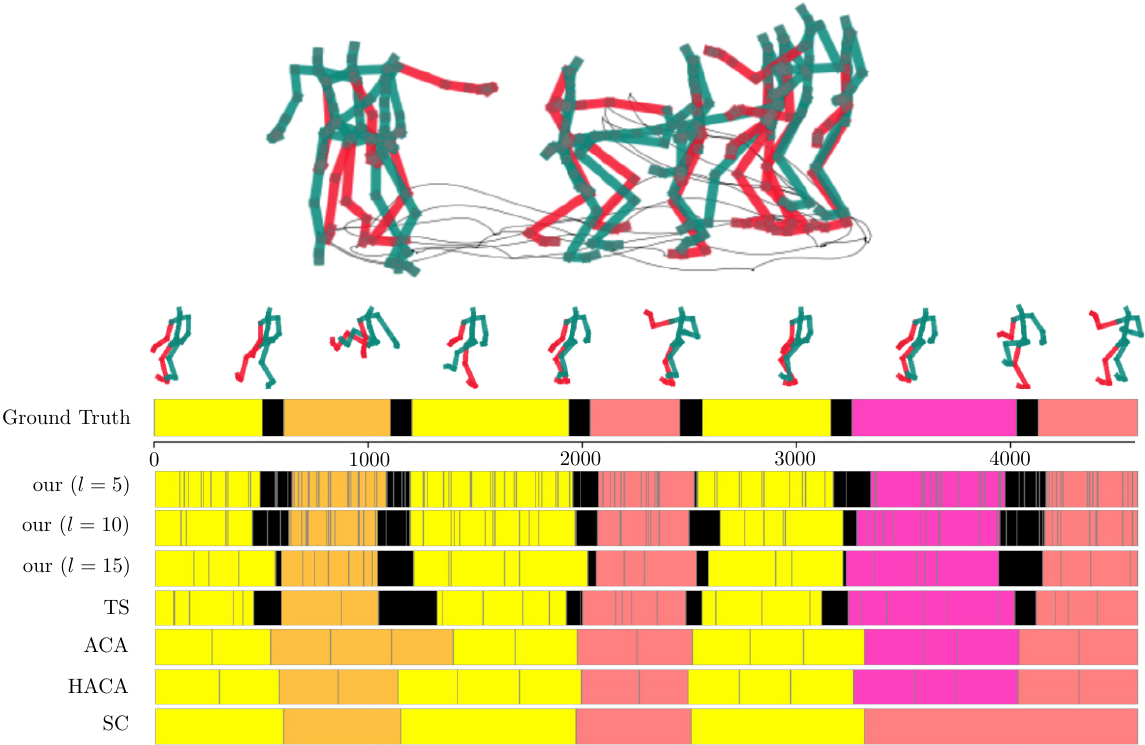}
    \caption{Trial \texttt{01} of subject \texttt{86} with its ground truth labels. Colors correspond to the labels \texttt{walk} (yellow), \texttt{jump} (orange), \texttt{punch} (light red), \texttt{kick} (magenta) and \texttt{transition} (black). Beneath the resulting labeling using our techniques with different values for $l$, temporal segmentation (TS), (hierarchical) aligned cluster algorithm (HACA/ACA) and spectral clustering (SC) \cite{krüger2017Efficient,zhou2012hierarchical,zhou2008aligned}, with gray lines corresponding to the start/end of frequent patterns.}
    \label{cmusub:1}
\end{figure}
\begin{figure}
    \scalebox{0.65}{\boxplot{experiments/data/accuracies.csv}{Accuracy}}
    \scalebox{0.65}{\boxplot{experiments/data/precision.csv}{Macro Precision}}
    \scalebox{0.65}{\boxplot{experiments/data/recall.csv}{Macro Recall}}\\
    \scalebox{0.625}{\drawChartB[log]{Time in $[s]$}{experiments/data/times.csv}[0.1cm][200]}
    \caption{Quantative analysis on trial \texttt{1} to \texttt{14} of subject \texttt{86} from \cite{mocap2007carnegie} and comparison of our techniques to temporal segmentation (TS), (hierarchical) aligned cluster algorithm (HACA/ACA) and spectral clustering (SC) from \cite{krüger2017Efficient,zhou2012hierarchical,zhou2008aligned}.}
    \label{fig:CMUAccuracy}
\end{figure}

Motion Segmentation finds applications in many different fields such as robotics, sports analysis or traffic monitoring~\cite{surveyMotionData}. We apply our techniques to this problem on the CMU data set \cite{mocap2007carnegie}. This data set consists of motion tracking data of $31$ different joint-trackers on different subjects doing sports (trials) ranging through different activities (refer to Figure~\ref{cmusub:1}). We interpret these as trajectories in $93$-dimensional Euclidean space by concatenating the three-dimensional coordinates of all joints back to back to form a pose. Each trial consists of up to $10^4$ poses. We then apply our subtrajectory clustering algorithm with $\Delta_\textit{simp}=0.8$ and $\Delta_\textit{free}\approx1.35$ and a complexity bound $l$ between $5$ and $15$, where the exact parameters have been identified via an exhaustive search to yield the best accuracy for the given complexity bound $l$. The output consists of a set of curves that act as cluster centers. For each of these centers we identify the ground truth label that best corresponds to this center and label all points in its $\Delta_\textit{free}$-Coverage with the identified label. Whenever different labels are assigned to a point along the curve we mark it as a transition between motions. 

\subparagraph{Evaluation} The resulting labeling can be seen in Figure~\ref{cmusub:1}. Observe in particular that an increase in $l$ decreases the number of patterns identified with the total number of labeled segments approaching that of the state-of-the-art, while the accuracy decreases only slightly. We compute the accuracy of the resulting segmentations on ground truth data from \cite{krüger2017Efficient} and compare this accuracy with the accuracy of the temporal segmentation approach (TS) discussed in \cite{krüger2017Efficient} as well as the aligned cluster algorithm (ACA), hierarchical aligned cluster algorithm (HACA) and spectral clustering (SC) discussed in~\cite{zhou2012hierarchical,zhou2008aligned}. The resulting accuracies can be seen in Figure~\ref{fig:CMUAccuracy}. The quantitative accuracy of our techniques compares well to the state-of-the-art techniques, with a (roughly) tenfold improvement in the running time.

\section{Discussion}

We observe that in practice the algorithm is much faster and yields better solutions than what could be expected from the theory. We partially attribute this to the fact that unlike in the worst case analysis the number of non-empty cells in the free-space is less than $n^2$. This reduces the number of extremal points and the complexity of the computed coverage. This suggests that analyses with additional input assumptions such as $c$-packedness \cite{DBLP:journals/dcg/DriemelHW12} or $\lambda$-low-density \cite{DBLP:journals/algorithmica/BergSVK02} could result in a theoretically founded explanation of the observed running time.
We further observe that the extension to non-constant complexity center curves indeed allows the algorithm to capture more interesting behaviour compared to when the center curves are restricted to constant/lower-complexity center curves. 

This provides evidence, that the problem formulation by \cite{Akitaya2021Covering} is practically viable and serves as a versatile tool analyse large amounts of spatio-temporal data.

\paragraph*{Acknowledgements}
This work was partially funded by the Deutsche Forschungsgemeinschaft (DFG, German Research Foundation) - 313421352 (FOR 2535 Anticipating Human Behavior) and the iBehave Network: Sponsored by the Ministry of Culture and Science of the State of North Rhine-Westphalia. The authors are affiliated with Lamarr Institute for Machine Learning and Artificial Intelligence. We thank Frederik Brüning for contributions in early stages of this research. We thank Jürgen Gall, Julian Tanke, Jürgen Kusche, and Bernd Uebbing for useful discussions on the data sets and real world problems. Special thanks to Simon Bartlmae and Paul Jünger for their assistance in conducting the experiments.

\bibliography{mybib}
\bibliographystyle{plainurl}

\end{document}